\documentclass{iac}
\usepackage{placeins}
\usepackage{pifont}
\usepackage{soul}
\usepackage{algorithm}
\usepackage{algorithmic}
\usepackage{tikz}
\usetikzlibrary{arrows.meta,positioning,calc}
\definecolor{revisionpink}{RGB}{220,20,120}

\newtheorem{lemma}{Lemma}[section]

\hypersetup{hidelinks}
\begin{document}

\IACconference{77}
\IAClocation{Antalya, Türkiye}
\IACdates{5--9 Oct}
\IACyear{2026}
\IACpapernumber{B2,5,8,x113206}
\IACcopyright{2026}{the International Astronautical Federation (IAF)}

\title{Goal-Oriented Bundle Management and Flow Control\\ for Deep-Space
Communications via Reinforcement Learning
}

\IACauthor*{Sajjad Baghaee}{}{}

\IACauthor*{Aimin Li}{}{}

\IACauthor{Shanmugarajan Balasubramanian}{}{}

\IACauthor{Elif Uysal}{}{}

\IACsharedaffiliation
{Communication Networks Research Group (CNG), Dept. of Electrical and
Electronics Engineering, Middle East Technical University (METU), Ankara,
Türkiye}
{\textsuperscript{a}\uline{sajjad@baghaee.com},
\textsuperscript{b}\uline{aimin@metu.edu.tr},
\textsuperscript{c}\uline{e279039@metu.edu.tr},
\textsuperscript{d}\uline{uelif@metu.edu.tr}}
\abstract{%
Deep-space communications operate under intermittent connectivity, long propagation delays, and relay storage constraints. To span vast distances, data are often forwarded through relay nodes, where incoming traffic can vary over time. In such networks, the Bundle Protocol encapsulates data into store-carry-forward units called bundles, which are stored at intermediate nodes and forwarded only when contact opportunities become available. Therefore, relay-side buffer and flow management become critical challenges, since limited storage and scarce contact windows force the relay to decide which bundles to transmit or discard. This decision process is inherently goal-oriented: the relay should not merely forward data, but should allocate limited transmission opportunities to the bundles that are most valuable to the downstream mission objective. Existing schedulers, such as FIFO, are heuristic and do not explicitly optimize task utility; moreover, they often lead to unfair service across heterogeneous planetary sources. This paper presents a closed-loop modeling and learning framework for relay-side bundle and flow management in a contact-driven deep-space network, where updates generated by lunar and Martian sources are forwarded via a relay to Earth. These bundles are stored in a shared relay buffer and carry metadata including source and destination identities, generation time, and time-to-live. On Earth, data utility depends on freshness, which we quantify via the excess age $\Delta_i^{e}(t) = {\left[\Delta_i(t) - d_i(t)\right]}^+$, where $\Delta_i(t)$ denotes the Age of Information (AoI) of source $i$, i.e., the time elapsed since the freshest update currently available on Earth was generated, and $d_i(t)$ is a source-dependent delay baseline. Based on this metric, we define a goal-oriented system loss that penalizes excess age in a source-dependent and nonlinear manner, allowing the relay to prioritize information according to mission relevance. We formulate relay control as a constrained semi-Markov decision process (CSMDP) and solve it using masked Lagrangian Proximal Policy Optimization (PPO), where feasibility-aware action masking restricts decisions to transmissions that can finish within the remaining contact window. Simulation results in DTNSim indicate that the proposed approach significantly improves age-based utility and achieves substantially better inter-flow fairness compared to baseline schedulers. By explicitly prioritizing information according to downstream task relevance and using this signal to guide transmission and dropping decisions, the proposed framework improves goal-oriented information delivery in deep-space networks and, to our knowledge, is among the first reinforcement learning frameworks to jointly optimize relay-side bundle scheduling and dropping, as well as flow control under contact-driven dynamics with explicit fairness considerations across heterogeneous sources.\vspace{-3mm}}

\maketitle

\section{Introduction}
\label{sec:introduction}

Deep-space networks rely on intermittent contacts and store-carry-forward relaying to communicate across long propagation delays. Delay/Disruption-Tolerant Networking (DTN) and the Bundle Protocol (BP) enable store-carry-forward communication in such environments \cite{rfc4838,rfc9171}, but leave a fundamental relay-side allocation problem unresolved: when multiple buffered bundles compete for a short contact window, which should be transmitted, and which should be dropped under buffer pressure. This decision becomes challenging in heterogeneous deep-space networks, where lunar and Martian traffic share a finite relay buffer but differ significantly in contact frequency and propagation delay. Simple queueing policies such as FIFO, LIFO, or fixed priority are therefore inadequate, as they may either waste scarce contact opportunities on outdated information or systematically disadvantage one source. To reason about this trade-off, we consider destination-side information freshness. Age of Information (AoI) captures the time elapsed since the most recently received update was generated \cite{aoi_foundations,aoi_multiple_sources}. However, in deep-space settings, raw AoI conflates two effects: the relay’s scheduling decisions and the underlying contact structure of each source. To separate these effects, we measure freshness relative to a contact-aware baseline di(t) and define the excess age $\Delta_i^{e}(t)=[\Delta_i(t)-d_i(t)]^+$. This allows us to penalize only the staleness attributable to relay decisions, while avoiding bias toward inherently disadvantaged sources.


To address this problem, we develop a relay-side scheduling and dropping model for deep-space DTNs, where bundle decisions are driven by downstream information utility under finite storage and intermittent contacts. We consider an Earth–Relay–Moon–Mars network with a shared relay buffer. The relay operates in an event-driven manner: at each decision epoch, it may transmit a feasible bundle, {resolve buffer overflow through rejection or eviction}, or remain idle. Since freshness at the destination is only updated upon Earth delivery, {pending deliveries are retained in the system state so that the delayed reception dynamics remain Markovian.} These dynamics lead to a constrained semi-Markov decision process (CSMDP), which we solve using masked Lagrangian Proximal Policy Optimization (PPO) with permutation-invariant encoding of the variable-size buffer state.

{Our main \textbf{contribution} is a goal-oriented framework for deep-space DTNs that jointly controls relay scheduling and dropping based on downstream information utility rather than queue order. It combines: (i) a closed-loop model with finite relay storage, scheduled contacts, stochastic source traffic, heterogeneous bundle sizes, and delayed destination updates; (ii) a contact-aware excess-age objective with explicit inter-source fairness constraints; (iii) an event-driven CSMDP formulation for joint scheduling and dropping; and (iv) masked Lagrangian PPO with permutation-invariant encoding of the variable-size relay buffer. The framework is implemented in DTNSim.}
\vspace{-2mm}
\section{Background and Related Work}
\subsection{Relay-Local DTN Resource Management}
DTN and the Bundle Protocol (BP) enable persistent store-carry-forward
communication across disrupted links while leaving relay-local bundle ordering
and discard policies unspecified \cite{rfc4838,rfc9171,caini2011satdtn}.
In scheduled space networks, Contact Graph Routing (CGR) and related methods
exploit contact plans to construct time-dependent routes
\cite{araniti2015cgr,fraire2021cgrtutorial}, with extensions addressing
heterogeneous traffic, contact capacity, and finite buffers
\cite{madoery2023heterogeneous,alhajj2024cgr}. These mechanisms primarily determine where and when bundles can be forwarded across the network.

Our focus is complementary: once multiple bundles coexist at a finite-buffer relay, the node must decide which bundle should use the next outgoing contact and which bundle should be discarded when storage becomes scarce. Such relay-local decisions can influence end-to-end performance under limited contact capacity \cite{LE2026108432}. Conventional rules such as FIFO, LIFO,
fixed-priority, and TTL-based policies are simple but do not account for the
freshness already available at the destination, motivating freshness-aware
joint scheduling and dropping {based on the downstream value of each candidate transmission}.

\vspace{-2mm}
\subsection{Freshness and Goal-Oriented Utility}

Age of Information (AoI) provides a natural measure of destination-side information freshness, and AoI-aware scheduling can substantially affect information timeliness
\cite{xie2024aoi_scheduling,tripathi2023freshcsma}. Related work has also examined fairness in freshness-oriented systems, since minimizing an aggregate age metric can produce highly uneven service across heterogeneous sources \cite{jones2024aoifairness,zhu2026fairnessaoi}.

Deep-space networks introduce an additional \textit{asymmetry}: sources may experience
substantially different propagation delays and contact opportunities even under
the same relay policy. We therefore evaluate freshness relative to a policy-independent reference induced by the known contact plan, and combine this relative freshness objective with explicit source-specific service constraints. This is consistent with goal-oriented
communication, where resources are allocated according to downstream utility
rather than delivered data volume alone \cite{li2024goal_framework,li2024goal_tensor}.

\vspace{-3mm}
\subsection{Learning-Based Control and Research Gap}

Learning-based methods have been studied for communication scheduling, routing,
and resource management, including constrained satellite routing and PPO-based
congestion control in deep-space DTNs
\cite{lyu2024constrainedrouting,yang2024deepdtn}. The central difficulty in our setting is temporal: a relay must decide which
bundle to transmit or discard using its current buffer state, while the value
of that decision is realized only when the transmitted information
reaches Earth. During this propagation interval, new bundles may arrive, others
may expire, contact conditions may change, and subsequent buffer-management
decisions may have been made. Relay actions are therefore coupled across
irregular decision epochs through delayed destination-side freshness.

This structure calls for an event-driven control formulation that preserves
the delayed effect of relay actions and the changing set of feasible decisions.
We formulate the problem as a constrained SMDP, where event-driven transitions
capture irregular decision times and delayed destination updates,
feasibility-aware action masking removes transmissions that cannot be completed
within the remaining contact window
\cite{huang2022masking}, and constrained policy optimization incorporates
operational and source-specific service requirements
\cite{achiam2017cpo,sun2024constrained}.
The formulation links relay-side scheduling and dropping decisions to
the freshness realized at Earth under stochastic, heterogeneous
traffic.

\vspace{-2mm}
\section{Deep-Space DTN System Model}
\label{sec:system}

\subsection{Network and Contact Model}

We consider a relay-side scheduling and dropping model for deep-space DTNs, where limited relay storage and intermittent contact opportunities require goal-oriented bundle management. As a representative scenario, we consider four nodes: Earth $E$, relay $R$, lunar source $M$, and Martian source $S$, with source set $\mathcal{I}\triangleq\{M,S\}$. Both sources generate time-stamped status-update bundles destined for Earth. Direct source--Earth delivery is excluded so that every bundle traverses the finite-buffer relay, the only node at which scheduling and dropping are controlled. Figure~\ref{fig:topology} shows this representative topology.

\begin{figure}
\centering
\includegraphics[width=\columnwidth]{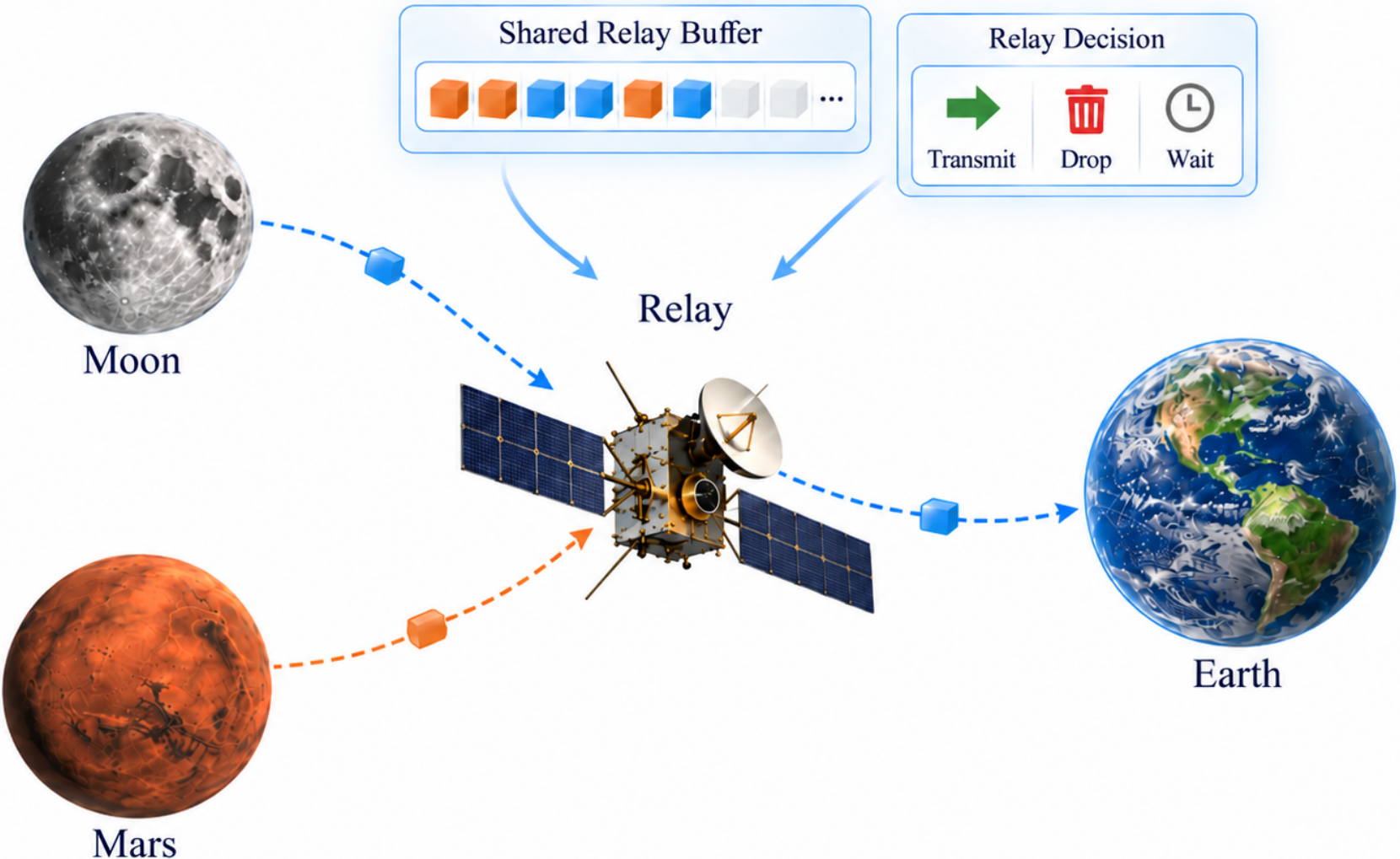}
\caption{Earth--Relay--Moon--Mars network.}
\label{fig:topology}
\end{figure}

Connectivity is specified by contact plans
\begin{equation}
\mathcal{C}_{xy}\triangleq
\left\{
\left(
t_{xy,s}^{(q)},t_{xy,e}^{(q)},
\sigma_{xy}^{(q)},\tau_{xy}^{(q)}
\right)
\right\}_{q=1}^{K_{xy}},
\end{equation}
where $x$ and $y$ denote link endpoints, $q$ indexes contact segments, and
$K_{xy}$ is the number of segments on $x\!\rightarrow\!y$; each segment
specifies its start/end times, usable rate, and one-way propagation delay.
We use the links $M\!\rightarrow\!R$, $S\!\rightarrow\!R$, and
$R\!\rightarrow\!E$. During active relay--Earth segment $q$,
$\sigma_{RE}(t)=\sigma_{RE}^{(q)}$,
$\tau_{RE}(t)=\tau_{RE}^{(q)}$, and
$t_{RE}^{end}(t)=t_{RE,e}^{(q)}$; otherwise
$\sigma_{RE}(t)=\tau_{RE}(t)=0$ and
$t_{RE}^{end}(t)=t$. Thus, the downlink is available iff
$\sigma_{RE}(t)>0$.

\subsection{Source Generation and Source--Relay Transport}
\label{subsec:source_transport}

Bundles are generated independently at the Moon and Mars according to Poisson
processes with heterogeneous inter-generation times and randomly drawn sizes. Each source maintains a non-preemptive FIFO queue; bundles accumulate during
contact gaps and are transmitted head-of-line when a source--relay contact is
active. A bundle is admitted to service only if its full serialization fits
within the remaining contact interval.


For source $i\in\mathcal I$, if bundle $b$ begins transmission at
$t_i^{tx,s}(b)$ over contact segment $q$, its source-link completion and relay
arrival times are
\begin{equation}
t_i^{tx,e}(b)
=
t_i^{tx,s}(b)
+\frac{L(b)}{\sigma_{iR}^{(q)}},
\qquad
t^{arr}(b)
=
t_i^{tx,e}(b)+\tau_{iR}^{(q)},
\end{equation}

with the feasibility condition
\begin{equation}
t_i^{tx,e}(b)\leq t_{iR,e}^{(q)}.
\end{equation}
Thus, only transmission completion must occur before the contact closes;
relay reception may occur later due to propagation delay.

\subsection{Bundle and Buffer Model}

Each bundle is represented as
\begin{equation}
b=
\bigl(
i(b),g(b),L(b),T^{ttl}(b),
t^{arr}(b),\omega(b),id(b)
\bigr),
\label{eq:bundle_representation}
\end{equation}
where $i(b)$ is the source, $g(b)$ the update-generation time, $L(b)$ the
size, $T^{ttl}(b)$ the lifetime, $t^{arr}(b)$ the relay-arrival time,
$\omega(b)$ an optional importance weight, and $id(b)$ a unique identifier.
Its expiry time is
$t^{exp}(b)\triangleq g(b)+T^{ttl}(b)$, and it is valid at time $t$ only
if $t<t^{exp}(b)$.

Let $Q_i(t)$ be the multiset of buffered source-$i$ bundles. The
shared relay buffer is
$\mathcal{B}(t)\triangleq\biguplus_{i\in\mathcal I}Q_i(t)$, with occupancy
$B(t)\triangleq\sum_{b\in\mathcal B(t)}L(b)\leq B_{\max}$.
An arriving bundle is admitted when capacity permits; otherwise the controller
must reject it or evict a bundle so that the post-decision occupancy
does not exceed $B_{\max}$. We assume each bundle fits within the relay buffer,
i.e., $L(b)\leq B_{\max}$.

\vspace{-1mm}
\subsection{Relay--Earth Service and Feasibility}

The relay--Earth downlink is a non-preemptive single server. Let
$busy(t)\in\{0,1\}$ indicate whether a transmission is active,
$b^{srv}(t)$ the in-service bundle, and $t^{cmp,srv}(t)$ its completion
time, fixed when transmission begins.
When the server is busy, $\tau^{srv}_{RE}(t)$ denotes the propagation delay
of the relay--Earth contact segment used for the current transmission and
remains fixed until completion. When the server is idle,
$b^{srv}(t)=\emptyset$, $t^{cmp,srv}(t)=+\infty$, and
$\tau^{srv}_{RE}(t)=0$. For a bundle $b$ selected at time $t$, its transmission, relay-completion,
and Earth-delivery times are

\begin{equation}
\begin{aligned}
t^{tx}(b,t)&=\frac{L(b)}{\sigma_{RE}(t)},\\
t^{cmp}(b,t)&=t+t^{tx}(b,t),\\
t^{del}(b,t)&=t^{cmp}(b,t)+\tau_{RE}(t),
\end{aligned}
\end{equation}
with $t^{tx}(b,t)=+\infty$ when $\sigma_{RE}(t)=0$.
When bundle $b$ starts service at time $t$,
$t^{cmp,srv}(t^+)=t^{cmp}(b,t)$.

We use a non-preemptive, all-or-nothing service model: once transmission begins, the bundle must complete within the same active contact and cannot be resumed later. Accordingly, bundle $b$ is schedulable at time $t$ only if
\begin{equation}
t^{cmp}(b,t)\leq t_{RE}^{end}(t),
\qquad
t^{del}(b,t)<t^{exp}(b).
\end{equation}
The first condition ensures transmission completion before contact termination,
while the second requires Earth reception before bundle expiry.

\vspace{-1mm}
\subsection{Pending Deliveries and Observability}

Propagation separates relay transmission completion from Earth reception.
We retain the pending-delivery multiset
$\mathcal{G}(t)\subseteq
\mathbb{R}_{\geq0}\times\mathcal I\times\mathbb{R}_{\geq0}$,
where $(\hat t,\hat\imath,\hat g)$ records an Earth-delivery time,
source, and generation timestamp. Completion of bundle $b$ inserts
$(t^{del}(b,t),i(b),g(b))$ into $\mathcal G$; at $\hat t$, the record is
removed and Earth-side freshness is updated as defined in
Subsection~\ref{subsec:earth_side_aoi}. Retaining $\mathcal G(t)$ preserves
the Markov state under propagation delay.

Under the {deterministic contact plan and transmission/propagation dynamics}, the relay can reconstruct
pending deliveries and Earth-side timestamps from its prior actions and the
known contact plan; instantaneous Earth feedback is therefore unnecessary.
{The policy observes only relay-available information and
the known contact plan.}
The SMDP state is defined in
Section~\ref{sec:event_driven_smdp}, while the policy feature
representation is defined in
Subsection~\ref{subsec:bundle_token_features}.

\section{Closed-Loop Freshness and Excess-Age Objective}
\label{sec:freshness_objective}

Relay-side transmission affects freshness only when the corresponding pending
delivery reaches Earth. We therefore define destination-side Age of Information
(AoI) and compare it with a source-dependent baseline derived from the known
contact plan.

\vspace{-1mm}
\subsection{Earth-Side Freshness}
\label{subsec:earth_side_aoi}

For each source $i\in\mathcal{I}$, let $u_i(t)$ denote the generation
timestamp of the freshest source-$i$ update delivered to Earth by time $t$.
When a pending delivery
$(\hat t,\hat{\imath},\hat g)\in\mathcal{G}(t^-)$ reaches Earth at
$t=\hat t$,
{$u_{\hat{\imath}}(t)\leftarrow
\max\{u_{\hat{\imath}}(t^-),\hat g\}$.
The AoI is $\Delta_i(t)\triangleq t-u_i(t)$.}
The maximum prevents an out-of-order delivery from replacing fresher
information at Earth. Relay transmission completion does not update
$u_i(t)$; only Earth reception can advance this timestamp, while
$\Delta_i(t)$ increases linearly between useful deliveries.
\vspace{-1mm}
\subsection{Contact-Plan Baseline}
Raw AoI reflects both relay-controlled effects and source-dependent contact and
propagation conditions. To account for this heterogeneity, we define a
policy-independent baseline based on idealized, contention-free forwarding
over the known contact plan.

For an $L$-bit bundle available at node $x$ at time $s$, its earliest feasible
arrival at node $y$ is

\begin{equation}
\begin{gathered}
\mathsf{A}_{xy}(s;L)
\triangleq\\[-0.6mm]
\min_{\substack{1\leq q\leq K_{xy}:\\
\eta_{xy}^{(q)}(s)+L/\sigma_{xy}^{(q)}
\leq t_{xy,e}^{(q)}}}
\left[
\eta_{xy}^{(q)}(s)
+\frac{L}{\sigma_{xy}^{(q)}}
+\tau_{xy}^{(q)}
\right],
\end{gathered}
\label{eq:earliest_link_arrival}
\end{equation}
where
$\eta_{xy}^{(q)}(s)\triangleq
\max\{s,t_{xy,s}^{(q)}\}$ and $\mathsf{A}_{xy}(s;L)=+\infty$ if no contact can carry the complete bundle. 

Let $L^{\mathrm{ref}}$ denote a fixed reference bundle size, common to all
sources and policies. For a hypothetical reference update generated at time
$g$, define
\vspace{-1.5mm}
\begin{equation}
\begin{aligned}
a_i(g)
&\triangleq
\mathsf{A}_{RE}\!\left(
\mathsf{A}_{iR}(g;L^{\mathrm{ref}});
L^{\mathrm{ref}}
\right),\\
\underline{u}_i(t)
&\triangleq
\sup\{g\leq t:a_i(g)\leq t\},\\
d_i(t)
&\triangleq
t-\underline{u}_i(t).
\end{aligned}
\label{eq:physical_lower_bound}
\vspace{-1mm}
\end{equation}
Here, $d_i(t)$ represents the freshness achievable by the reference update
under idealized forwarding. Because it depends only on the contact plan and
$L^{\mathrm{ref}}$, it is independent of relay scheduling, queueing,
contention, dropping, and the realized traffic process.
\vspace{-1.5mm}
\subsection{Excess-Age Objective}

The excess age and instantaneous goal-oriented freshness cost are
\vspace{-1.2mm}
\begin{equation}
\begin{aligned}
\Delta_i^{e}(t)
&\triangleq
[\Delta_i(t)-d_i(t)]^+,\\
G(t)
&\triangleq
\sum\nolimits_{i\in\mathcal{I}}
w_i\phi_i\!\left(\Delta_i^{e}(t)\right),
\end{aligned}
\label{eq:instantaneous_system_cost}
\vspace{-2.2mm}
\end{equation}
where $[x]^+\triangleq\max\{x,0\}$, $w_i>0$ denotes the relative mission
importance of source $i$, and
$\phi_i:\mathbb{R}_{\geq0}\rightarrow\mathbb{R}_{\geq0}$ is nonnegative
and nondecreasing. Thus, $\Delta_i^{e}(t)$ measures Earth-side staleness
relative to the corresponding source-dependent contact-plan baseline.
Over an episode $[t_0,T]$, the relay minimizes the expected cumulative
freshness cost
{$J_G\triangleq\mathbb{E}\!\left[\int_{t_0}^{T} G(t)\,dt\right]$},
subject to the operational and source-specific constraints introduced later.
Its event-wise decomposition defines the SMDP reward in
Section~\ref{subsec:transition_kernel}.

\vspace{-1mm}
\vspace{-1mm}
\section{Event-Driven SMDP Formulation}
\label{sec:event_driven_smdp}

Relay-side scheduling and dropping occur at irregular times determined by contact changes, bundle arrivals, transmission completions, Earth deliveries, and expirations. We therefore model relay control as a semi-Markov decision process (SMDP) \cite{sutton1999smdp}. Let $\{t_k\}_{k\geq0}$ denote the strictly increasing
decision epochs and $\delta_k\triangleq t_{k+1}-t_k$ the holding time. Co-timed events are processed as one batch before control. If an arrival creates buffer overflow, the controller may take a finite sequence of rejection or eviction decisions at the same physical time until the overflow is resolved. A scheduling decision is then made when transmission is possible. These same-time control stages do not constitute separate SMDP transitions; physical time advances only after all applicable stages have been completed.

\vspace{-1.7mm}
\subsection{Events, State, and Action Sets}
\label{subsec:event_types_ordering}
\label{subsec:state_mode_action_sets}

The physical event types are \emph{contact boundary}, \emph{transmission completion},
\emph{Earth delivery}, \emph{relay expiry}, and \emph{relay arrival}.
Co-timed events are processed in this fixed order. This convention first
establishes the active contact state, then resolves transmission and destination
updates, removes expired bundles, and finally processes new arrivals, yielding a
unique post-event state. An additional controller wake-up may be scheduled after an elective wait, as described below; it changes no physical state by itself.
After processing events at $t_k$, the relay-side control state is
\vspace{-2mm}
\begin{align}
s_k\triangleq\Big(
&\mathcal{B}(t_k),B(t_k),
busy(t_k),b^{srv}(t_k),t^{cmp,srv}(t_k),\notag\\[-0.8mm]
&t_k,\sigma_{RE}(t_k),\tau^{srv}_{RE}(t_k),
\tau_{RE}(t_k),t^{end}_{RE}(t_k),
\label{eq:smdp_state}\\[-0.8mm]
&\{u_i(t_k)\}_{i\in\mathcal I},\mathcal{G}(t_k),
\{D_i(k)\}_{i\in\mathcal I},b_k^{new}
\Big).\notag
\vspace{-2mm}
\end{align}
The physical, service, and freshness variables are defined in
Sections~\ref{sec:system}--\ref{sec:freshness_objective}.
The variables $D_i(k)$ are the fairness deficits defined in
Section~\ref{sec:fairness_modeling}, while $b_k^{new}$ is nonempty only when
a newly arrived bundle requires an overflow decision. Equation~\eqref{eq:smdp_state} collects the relay-side control state from
which the policy observation is constructed. Let $h_k$ denote the upstream simulator state, including source
FIFO queues, active source--relay transmissions, and in-flight source--relay
bundles. The environment state is
$
x_k\triangleq (s_k,h_k).
$
The environment evolves Markovianly on $x_k$, whereas the relay policy acts only on information derived from $s_k$. The hidden upstream state $h_k$, together with source-generation and bundle-size draws, determines future relay arrivals but is not exposed
to the controller.

For scheduling, define the currently feasible bundle set
\begin{equation}
\begin{aligned}
\mathcal{F}(s_k)
\triangleq
\Big\{
b\in\mathcal{B}(t_k):
&t^{cmp}(b,t_k)\leq t^{end}_{RE}(t_k),\\[-0.7mm]
&t^{del}(b,t_k)<t^{exp}(b)
\Big\}.
\end{aligned}
\label{eq:feasible_schedule_set}
\vspace{-2mm}
\end{equation} 
When the relay--Earth server is idle,
the contact is active, and $\mathcal{F}(s_k)\neq\emptyset$, the scheduling action set is
\begin{equation}
\mathcal{A}^{sched}(s_k)
\triangleq
\mathcal{F}(s_k)\cup\{\mathrm{wait}\}.
\label{eq:scheduling_action_set}
\vspace{-2mm}
\end{equation}
Here, $\mathrm{wait}$ is an elective action. If $\mathcal{F}(s_k)=\emptyset$, no scheduling decision is required.

Overflow is resolved through one or more same-time drop stages. Let
$s_k^{d,0}\triangleq s_k$. Whenever a pending arrival
$b_k^{new}\neq\emptyset$ would violate the shared relay capacity, the
available drop actions at stage $j$ are
\begin{equation}
\mathcal{A}^{drop}(s_k^{d,j})
\triangleq
\{b_k^{new}\}\cup\mathcal{B}(s_k^{d,j}).
\label{eq:drop_action_set}
\vspace{-2mm}
\end{equation}
The drop stage is invoked only while admission of $b_k^{new}$ would violate
the shared relay-buffer capacity. Selecting $b_k^{new}$ rejects the arriving
bundle and terminates overflow resolution. Selecting
$b\in\mathcal{B}(s_k^{d,j})$ evicts that buffered bundle. If the remaining
capacity is still insufficient to admit $b_k^{new}$, another drop stage is
performed at the same physical time using the updated buffer state. Otherwise,
$b_k^{new}$ is admitted and the overflow sequence terminates.
For selected drop action $d_{k,j}$, define
$s_k^{d,j+1}
\triangleq
\Psi^{drop}(s_k^{d,j},d_{k,j}).$
Let $J_k$ denote the total number of rejection or eviction decisions executed
during overflow resolution at physical epoch $t_k$. Each such decision
discards one bundle, and $J_k=0$ when no overflow action is required. 
The resulting post-overflow state is
$
s_k^{d}\triangleq s_k^{d,J_k}
$,
with $s_k^{d}=s_k$ when $J_k=0$. 
The controller therefore executes state-dependent drop sequence
$d_{k,0},\ldots,d_{k,J_k-1}$, where each action is selected from the intermediate state produced by the
preceding stage. After overflow resolution, scheduling is reconsidered. If the
server is idle, $\sigma_{RE}(t_k)>0$, and
$\mathcal F(s_k^{d})\neq\emptyset$, then
$q_k\in\mathcal A^{\mathrm{sched}}(s_k^{d})$; otherwise
$q_k=\emptyset$. The complete control applied at epoch $t_k$ is
{$a_k\triangleq(\mathbf d_k,q_k)$, where
$\mathbf d_k=(d_{k,0},\ldots,d_{k,J_k-1})$ is the drop sequence,
empty when no overflow occurs.}
\vspace{-2mm}
\subsection{Post-Decision State and Next Event}
After the drop sequence has been completed, the scheduling component is
applied:
$\bar{s}_k
\triangleq
\Psi^{sched}(s_k^{d},q_k),$
where $\Psi^{sched}(s,\emptyset)=\Psi^{sched}(s,\mathrm{wait})=s$.
If
$q_k=b\in\mathcal{F}(s_k^{d})$, bundle $b$ is removed from the relay buffer
and the server state is set to
$busy=1$, $b^{srv}=b$,
$t^{cmp,srv}=t_k+t^{tx}(b,t_k)$, and
$\tau^{srv}_{RE}=\tau_{RE}(t_k)$.
No holding time, reward interval, or discount factor is introduced between
same-time drop stages or between overflow resolution and scheduling.

Let
$T_{\mathrm{exp}}$, $T_{\mathrm{ct}}$, $T_{\mathrm{tx}}$,
$T_{\mathrm{del}}$, and $T_{\mathrm{arr},k}$ denote, respectively, the next
buffered-bundle expiry, relay--Earth contact boundary, transmission completion,
pending Earth delivery, and relay arrival. The candidate times are
\begin{align}
&T_{\mathrm{exp}}(\bar{s}_k)
\triangleq
\min\left\{
t^{exp}(b):
b\in\bar{\mathcal B}_k,\;
t^{exp}(b)>t_k
\right\},
\notag\\
&T_{\mathrm{ct}}(\bar{s}_k)
\triangleq
\inf\left\{
t>t_k:
t\in
\bigcup_{q=1}^{K_{RE}}
\{t_{RE,s}^{(q)},t_{RE,e}^{(q)}\}
\right\},
\notag\\
&T_{\mathrm{tx}}(\bar{s}_k)
\triangleq
\begin{cases}
t^{cmp,srv}(\bar{s}_k),
&
\begin{aligned}
&busy(\bar{s}_k)=1,\\[-0.5mm]
&t^{cmp,srv}(\bar{s}_k)>t_k,
\end{aligned}
\\
+\infty,
&\text{otherwise},
\end{cases}
\label{eq:candidate_event_times}\\
&T_{\mathrm{del}}(\bar{s}_k)
\triangleq
\min\left\{
\hat t:
(\hat t,\hat i,\hat g)\in\bar{\mathcal G}_k,\;
\hat t>t_k
\right\}.
\notag
\vspace{-2mm}
\end{align}

Here, $\bar{\mathcal B}_k$ and $\bar{\mathcal G}_k$ denote the buffered-bundle
and pending-delivery multisets in the post-decision state $\bar{s}_k$.
A minimum or infimum over an empty candidate set is defined as $+\infty$.
Let $\bar{x}_k\triangleq(\bar{s}_k,h_k)$ denote the complete post-decision state. The next relay-arrival time is a simulator-side random variable whose conditional law is generated by upstream evolution from $\bar{x}_k$, including subsequent source-generation and bundle-size draws. Neither hidden upstream information nor future random outcomes are exposed to the policy. Thus, $T_{\mathrm{arr},k}>t_k$ denotes the first unprocessed
relay arrival strictly after the co-timed event batch, with $T_{\mathrm{arr},k}=+\infty$ if no further arrival occurs.

{
An elective scheduling wait introduces an additional decision timer
\vspace{-1mm}
\begin{equation}
T_{\mathrm{wake}}(\bar{s}_k,q_k)
\triangleq
\begin{cases}
t_k+30~\mathrm{s}, & q_k=\mathrm{wait},\\
+\infty, & \text{otherwise}.
\end{cases}
\label{eq:elective_wait_wakeup}
\vspace{-1mm}
\end{equation}
A physical event occurring before this timer cancels the need to wait for the
full $30$~s because the earlier event itself creates the next decision
opportunity. Forced idle periods, for which no feasible transmission exists,
do not create this timer and remain purely event-driven.
} For every nonterminal epoch $t_k<T$, the next decision epoch is
\vspace{-2mm}
\begin{equation}
\begin{aligned}
t_{k+1}
\triangleq
\min\Bigl\{&
T_{\mathrm{arr},k},
T_{\mathrm{exp}}(\bar{s}_k),
T_{\mathrm{ct}}(\bar{s}_k),
T_{\mathrm{tx}}(\bar{s}_k),\\
&
T_{\mathrm{del}}(\bar{s}_k),
{T_{\mathrm{wake}}(\bar{s}_k,q_k),}
T
\Bigr\}.
\end{aligned}
\vspace{-1mm}
\label{eq:next_event_time}
\end{equation}

Because all events and control operations occurring at $t_k$ are exhausted
before this minimization, every nonterminal finite event candidate is strictly
larger than $t_k$. Hence,
\vspace{-4mm}
\begin{equation}
t_{k+1}>t_k,
\qquad
\delta_k\triangleq t_{k+1}-t_k>0.
\label{eq:positive_holding_time}
\vspace{-2mm}
\end{equation}
If several future physical events attain the same minimum, they are processed
as one batch according to
Subsection~\ref{subsec:event_types_ordering}. {At the terminal
time $T$, physical events occurring exactly at $T$ are processed before the
episode stops. The simulator then terminates without draining the system;
unexpired bundles still queued, in service, or in transit are recorded as
unresolved backlog rather than as drops or expirations.}

\vspace{-1mm}
\subsection{Event Updates, Transition Kernel, and Reward}
\label{subsec:transition_kernel}

At $t_{k+1}$, only state components affected by the occurring physical events
are updated. A relay arrival is admitted if capacity permits; otherwise the
existing buffer is retained and the arrival becomes $b_{k+1}^{new}$, so that
the first control stage at the resulting epoch is overflow resolution.
{If one eviction does not release sufficient capacity, additional
drop stages are taken at the same timestamp until the arrival is either
admitted or rejected. Scheduling is then reconsidered without advancing
physical time.}

{
When a relay-buffered bundle reaches its absolute expiry time, it is removed
from the relay buffer. Expiry accounting is maintained according to the
bundle's physical location in the simulator.

A queued bundle that can no longer be delivered before its
expiry is excluded from the scheduling set through
Equation~\eqref{eq:feasible_schedule_set}, but it is not proactively removed;
it remains in the relay buffer until it is explicitly dropped or reaches its
absolute expiry time.}

At transmission completion, the completed bundle is not yet treated as
delivered because it must still propagate to Earth. Before clearing the server,
the record
$\left(t_{k+1}+\tau^{srv}_{RE}(t_{k+1}^{-}), i(b^{srv}(t_{k+1}^{-})), g(b^{srv}(t_{k+1}^{-})) \right)$
is inserted into $\mathcal{G}$. Here,
$\tau^{srv}_{RE}(t_{k+1}^{-})$ is the propagation delay of the contact segment
on which the completed transmission was initiated. The server is then cleared
by setting $busy=0$, $b^{srv}=\emptyset$, $t^{cmp,srv}=+\infty$, and
$\tau^{srv}_{RE}=0$.

An Earth-delivery event removes the corresponding pending record and updates
the Earth-side timestamp according to
{Subsection~\ref{subsec:earth_side_aoi}.} Thus, transmission completion and
Earth-side freshness update remain distinct physical events. After all events
at $t_{k+1}$ are processed, the fairness deficits are updated using the holding
time $\delta_k$ and the delivered freshness gains $y_{i,k}$ according to
Section~\ref{sec:fairness_modeling}.

{
The post-decision map, event scheduler, source-generation processes, and
transmission dynamics induce the controlled transition kernel on the complete
environment state
$P(dx'|x_k,a_k)\triangleq
\Pr\{x_{k+1}\in dx'\mid x_k,a_k\}$.

The relay policy selects actions using only the observation constructed from
$s_k$ and does not observe $h_k$ or future random outcomes. Conditional on a
fixed physical traffic trace, the environment transition is deterministic.
Under stochastic traffic, randomness originates from the independent
source-generation processes and bundle-size draws; source FIFO service,
source--relay contacts, and propagation then determine the resulting
relay-arrival sequence. Because physical time is included in $x_k$, the
transition kernel also induces the holding time
$\delta_k=t_{k+1}-t_k$.
}

The reward accumulated during transition $k$ is the negative integrated
freshness cost
{$r_k\triangleq-\int_{t_k}^{t_{k+1}}G(t)\,dt$.}
Integration over physical time correctly weights intervals of different
duration. Drop and expiry counts are retained as constraint costs rather than
being penalized in the reward.

{
Because multiple overflow actions may occur at one physical epoch, define
$N_k^{\mathrm{drop}}\triangleq J_k$,
where each rejected arriving bundle or evicted buffered bundle counts as one
drop. Let $N_k^{\mathrm{exp}}$ denote the total number of bundles whose absolute
lifetime expires during transition $k$; the simulator additionally records
the physical location at which each expiry occurs.}
Their constraints are defined in
Section~\ref{sec:constrained_smdp_formulation}.
\vspace{-2mm}
\section{Fairness Modeling}
\label{sec:fairness_modeling}

Minimizing aggregate excess age does not guarantee balanced service across
heterogeneous sources. We therefore measure fairness through \emph{useful
delivered freshness}: service is credited only when an Earth delivery advances
the latest generation timestamp of that source.

\vspace{-2mm}

\subsection{Freshness-Service Deficit}

Because $u_i(t)$ changes only when a useful source-$i$ update reaches Earth,
we define the delivered freshness gain during transition $k$ as
\vspace{-1mm}
\begin{equation}
y_{i,k}
\triangleq
u_i(t_{k+1})-u_i(t_k)
\geq 0.
\label{eq:delivered_freshness_gain}
\end{equation}
The nonnegativity follows directly from
{Subsection~\ref{subsec:earth_side_aoi}}, since the max-update rule makes
$u_i(t)$ nondecreasing.
This quantity equals the sum of the positive generation-timestamp advances
produced by all source-$i$ Earth deliveries during $(t_k,t_{k+1}]$; obsolete
or out-of-order deliveries contribute zero through
{Subsection~\ref{subsec:earth_side_aoi}.}
Hence, $y_{i,k}$ measures useful freshness progress rather
than the number or temporal regularity of delivered bundles.

Let $\rho_i>0$ denote the required long-run freshness-service rate of
source $i$, measured as delivered generation-time advance per unit physical
time. Thus, $\rho_i$ specifies a target rate of cumulative
freshness progress, not a required bundle-delivery frequency or a bound on
inter-delivery time.
Because transition $k$ lasts $\delta_k$, its fairness cost is
\begin{equation}
c_k^{(\mathrm{fair},i)}
\triangleq
\rho_i\delta_k-y_{i,k},
\label{eq:fairness_constraint_cost}
\end{equation}
where positive values indicate service shortfall. Accumulated under-service is
tracked by the nonnegative deficit
\begin{equation}
D_i(k+1)\triangleq
\left[D_i(k)+c_k^{(\mathrm{fair},i)}\right]^+ .
\label{eq:deficit_update}
\vspace{-2mm}
\end{equation}
By including $D_i(k)$ in the relay-side control state $s_k$, the policy can
condition scheduling and dropping decisions on accumulated source-specific
freshness-service shortfall.

The targets $\rho_i$ should be compatible with the achievable long-run
freshness-service rates under the contact and traffic dynamics. Excess service
can reduce an existing deficit, while the projection in
Equation~\eqref{eq:deficit_update} prevents surplus service from accumulating
as negative credit. Hence, $D_i(k)$ tracks uncompensated freshness-service
shortfall rather than imposing a bound on inter-delivery gaps.
\vspace{-2mm}
\subsection{Long-Run Fairness Interpretation}
Equation~\eqref{eq:deficit_update} is a virtual-queue recursion with demand
$\rho_i\delta_k$ and service $y_{i,k}$.
\vspace{-2mm}
\begin{lemma}[Deficit rate stability implies freshness fairness]
\label{lem:deficit_stability_fairness}
Let $\tilde{t}_K\triangleq\sum_{k=0}^{K-1}\delta_k=t_K-t_0$. If
$\tilde{t}_K\to\infty$ and
$D_i(K)/\tilde{t}_K\to0$, then
\vspace{-2mm}
\[
\liminf_{K\to\infty}
\frac{1}{\tilde{t}_K}
\sum_{k=0}^{K-1}y_{i,k}
\geq \rho_i .
\]
\end{lemma}

\begin{proof}
Since $[x]^+\geq x$, summing Equation~\eqref{eq:deficit_update} gives
$\sum_{k=0}^{K-1}y_{i,k}
\geq\rho_i \tilde{t}_K-D_i(K)+D_i(0)$; division by $\tilde{t}_K$ and deficit rate
stability yield the result.
\end{proof}
\vspace{-2mm}
For finite-horizon learning, $D_i(k)$ serves as a state variable that
exposes accumulated freshness-service shortfall to the policy. Separately,
$c_k^{(\mathrm{fair},i)}$ is used as the signed episodic constraint cost in
Section~\ref{sec:constrained_smdp_formulation}.
The lemma gives a conditional long-run guarantee on the average delivered
freshness-progress rate, provided the deficit is rate stable. It does not
guarantee bounded inter-delivery times, bounded AoI, or service within every
finite time window. The mechanism targets long-run freshness-service balance rather than fairness over every finite time window.
\vspace{-2.3mm}
\section{Constrained SMDP Problem Formulation}
\label{sec:constrained_smdp_formulation}

Using the event-driven SMDP defined in
Section~\ref{sec:event_driven_smdp}, we optimize destination-side freshness
subject to operational and source-specific fairness constraints.
\vspace{-2.1mm}
\subsection{Returns and Constraint Costs}

Let $\pi_\theta$ denote the stochastic relay policy applied at each
applicable control stage using the policy observation constructed from the
relay-side control state. For physical transition $k$, the complete same-time
control sequence is
$a_k\triangleq(\mathbf d_k,q_k)$, with
$\mathbf d_k=(d_{k,0},\ldots,d_{k,J_k-1})$,
as defined in Section~\ref{sec:event_driven_smdp}. The drop-stage decisions are
taken sequentially on their intermediate states, followed by the optional
scheduling decision $q_k$, while all such same-time stages belong to one
physical SMDP transition.

The expected finite-horizon reward and constraint returns are
\vspace{-1.7mm}
\begin{equation}
\begin{aligned}
J_R(\theta)
&\triangleq
\mathbb{E}_{\mu_0,\pi_\theta,P}
\left[
\sum_{k:t_k<T} r_k
\right],\\
J_{C_j}(\theta)
&\triangleq
\mathbb{E}_{\mu_0,\pi_\theta,P}
\left[
\sum_{k:t_k<T} c_k^{(j)}
\right],
\qquad j\in\mathcal{J},
\end{aligned}
\label{eq:reward_constraint_returns}
\vspace{-1mm}
\end{equation}
where $\mu_0$ denotes the distribution of the initial complete environment
state, and the reward $r_k$ and transition kernel $P$ are defined in
Subsection~\ref{subsec:transition_kernel}.

For relay management, the constraint index set and operational costs are
\begin{equation}
\begin{aligned}
\mathcal{J}
&\triangleq
\{\mathrm{drop},\mathrm{exp}\}
\cup
\{(\mathrm{fair},i):i\in\mathcal{I}\},\\
c_k^{(\mathrm{drop})}
&\triangleq
N_k^{\mathrm{drop}},
\qquad
c_k^{(\mathrm{exp})}
\triangleq
N_k^{\mathrm{exp}},
\end{aligned}
\label{eq:constraint_costs}
\vspace{-1.4mm}
\end{equation}
The fairness cost $c_k^{(\mathrm{fair},i)}$ is defined in
Equation~\eqref{eq:fairness_constraint_cost}, while
$N_k^{\mathrm{drop}}$ and $N_k^{\mathrm{exp}}$ are defined in
Subsection~\ref{subsec:transition_kernel}.
Each constraint $j\in\mathcal{J}$ has an episodic budget $d_j$.

Let $K_T$ denote the number of SMDP transitions in an episode ending at $T$,
so $t_{K_T}=T$. The undiscounted source-$i$ fairness cost then satisfies
\begin{equation}
\sum_{k=0}^{K_T-1}c_k^{(\mathrm{fair},i)}
=
\rho_i(T-t_0)
-
\left[u_i(T)-u_i(t_0)\right],
\label{eq:undiscounted_fairness_return}
\vspace{-2mm}
\end{equation}
where the equality follows from
$\sum_{k=0}^{K_T-1}\delta_k=T-t_0$ and the telescoping definition of
$y_{i,k}$ in Equation~\eqref{eq:delivered_freshness_gain}.

Thus, the episodic source-$i$ constraint regulates the total delivered
freshness progress over the episode. Because $c_k^{(\mathrm{fair},i)}$ is signed, excess freshness service during
some transitions can compensate for shortfall during others in the episodic
constraint return. This differs from the projected deficit recursion in
Equation~\eqref{eq:deficit_update}, which does not retain surplus service as
negative credit. This finite-horizon constraint is distinct from the conditional
long-run freshness-service guarantee in
Lemma~\ref{lem:deficit_stability_fairness}, which additionally requires
deficit rate stability.
\vspace{-2mm}
\subsection{Constrained Objective and Lagrangian}
The relay-control problem is
\begin{equation}
\begin{aligned}
\max_{\theta}\quad
&J_R(\theta)\\
\mathrm{s.t.}\quad
&J_{C_j}(\theta)\leq d_j,
\qquad j\in\mathcal{J}.
\end{aligned}
\label{eq:constrained_optimization_problem}
\vspace{-1mm}
\end{equation}
Introducing a nonnegative multiplier $\lambda_j\geq0$ for each constraint,
we form the Lagrangian
\begin{equation}
\mathcal{L}(\theta,\boldsymbol{\lambda})
\triangleq
J_R(\theta)
-
\sum_{j\in\mathcal{J}}
\lambda_j
\left(
J_{C_j}(\theta)-d_j
\right),
\label{eq:lagrangian}
\vspace{-1.2mm}
\end{equation}
where $\boldsymbol{\lambda}\triangleq(\lambda_j)_{j\in\mathcal{J}}$ and the primal--dual problem is
\vspace{-1.2mm}
\begin{equation}
\max_{\theta}\;
\min_{\boldsymbol{\lambda}\succeq0}
\mathcal{L}(\theta,\boldsymbol{\lambda}).
\label{eq:primal_dual_problem}
\vspace{-1.2mm}
\end{equation}
Section~\ref{sec:masked_lagrangian_ppo} implements this formulation using
masked Lagrangian PPO with projected multiplier updates.

\section{Method: Masked Lagrangian PPO with Variable-Cardinality Set Pooling}
\label{sec:masked_lagrangian_ppo}

We develop a masked Lagrangian PPO controller ~\cite{schulman2017ppo} for the constrained event-driven relay-control problem. The method combines stage-specific
physical-feasibility masking, a causal observation adapter for
variable-cardinality bundle candidates, permutation-invariant set pooling,
holding-time-aware advantage estimation, and Lagrangian constraint updates.
Overflow resolution may require multiple same-time drop-stage decisions
followed by an optional scheduling decision, while PPO optimization is
performed at the level of physical SMDP transitions.
Fig.~\ref{fig:masked_lagrangian_ppo} summarizes the controller.


\begin{figure*}[t]
\centering
\includegraphics[width=1\textwidth]{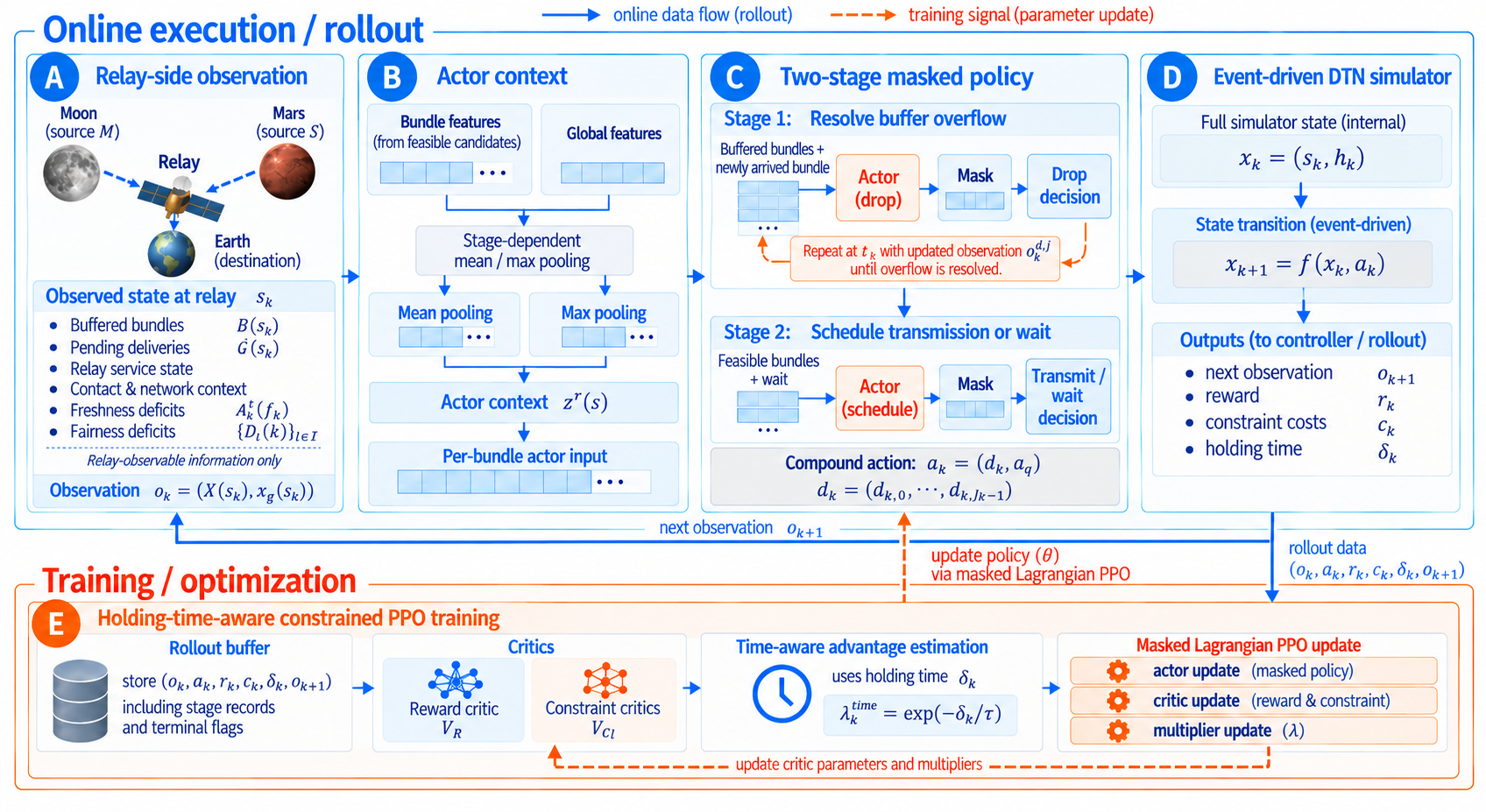}
\caption{Masked Lagrangian PPO architecture for event-driven relay control.}
\label{fig:masked_lagrangian_ppo}
\vspace{-2mm}
\end{figure*}

\vspace{-2mm}
\subsection{Stage-Specific Feasibility Masking}
\label{subsec:action_masking}

Let $s$ denote the relay-side control state presented to an individual
control stage. The nominal candidate sets are
\begin{equation}
\begin{aligned}
\widetilde{\mathcal A}^{\mathrm{sched}}(s)
&\triangleq
\mathcal B(s)\cup\{\mathrm{wait}\},
\\
\widetilde{\mathcal A}^{\mathrm{drop}}(s)
&\triangleq
\mathcal B(s)\cup\{b^{\mathrm{new}}(s)\}.
\end{aligned}
\label{eq:nominal_action_sets}
\vspace{-1mm}
\end{equation}
The scheduling slate contains the currently buffered bundles together with
the elective $\mathrm{wait}$ action, whereas the drop slate contains the
currently buffered bundles together with the pending arrival. The drop stage
is active only when the pending arrival requires overflow resolution.
For stage $r\in\{\mathrm{drop},\mathrm{sched}\}$, define
\begin{equation}
\chi^r(u;s)
\triangleq
\mathbf 1\!\left\{
u\in\mathcal A^r(s)
\right\},
\qquad
u\in\widetilde{\mathcal A}^{r}(s),
\label{eq:feasibility_mask}
\vspace{-1mm}
\end{equation}
where $\mathcal A^r(s)$ is the corresponding feasible action set defined in
Section~\ref{sec:event_driven_smdp}.

During overflow resolution, every represented bundle in the current drop
slate is eligible for rejection or eviction, while the special wait action is
not available. After each selected rejection or eviction, the relay state,
candidate set, observation, and mask are recomputed on the resulting
intermediate state $s_k^{d,j}$. If the incoming bundle remains pending and
the capacity violation persists, another masked drop decision is taken at
the same physical time. Once the arrival is admitted or rejected, scheduling
is reconsidered on the resulting state $s_k^d$. No physical time elapses
between these same-time control stages. At a scheduling stage, the $\mathrm{wait}$ action is available and a bundle
is enabled only if relay--Earth service is available, its complete
serialization fits within the active contact, and its Earth reception occurs
strictly before its absolute expiry, per
Equation~\eqref{eq:feasible_schedule_set}. If no scheduling stage is enabled,
the relay is forced idle and no scheduling action is sampled. The same
physical masks are used during training and execution; deterministic
evaluation selects the masked argmax.

\subsection{Causal Observation \& Set Pooling}
\label{subsec:bundle_token_features}
\label{subsec:observation_set_pooling}

The policy does not observe the hidden upstream state $h_k$ introduced in
Section~\ref{sec:event_driven_smdp}. Instead, a deterministic adapter
constructs a causal numerical observation from the relay-side control state
$s$. For a control stage with $n(s)$ represented bundle candidates,
\vspace{-0.5mm}
\begin{equation}
o(s)
\triangleq
\left(
X(s),x_g(s)
\right), X(s)\in\mathbb R^{n(s)\times18}, x_g(s)\in\mathbb R^{50},
\label{eq:policy_observation}
\end{equation}

where each row of $X(s)$ is the 18-dimensional representation of one
represented bundle and $x_g(s)$ is the 50-dimensional raw global context.
Stable bundle identifiers are retained for action decoding but are not part
of these numerical feature tensors. The same feature definitions and fixed
normalization rules are used during training, validation, and testing.

Each bundle is represented by 18 causal features grouped into four categories:
source and size attributes, freshness relevance, expiry/contact feasibility,
and buffer-management context. These features include bundle age and remaining
lifetime, potential freshness gain, transmission and deadline slack, overflow
pressure, and contact-opportunity indicators.
The 50-dimensional global vector is naturally partitioned as
\begin{equation}
x_g(s)
=
\left[
x_{\mathrm{core}}(s),
x_{\mathrm{srv}}(s),
x_{\mathrm{pend}}(s),
x_M^{\mathrm{src}}(s),
x_S^{\mathrm{src}}(s)
\right],
\label{eq:global_feature_blocks}
\end{equation}
with dimensions $17+6+5+11+11=50$.
The 50-dimensional global context summarizes five information groups:
the current relay/contact state, active service, pending Earth delivery,
and source-specific recent ingress/contact history for the Moon and Mars.
All components are constructed from information available at the relay and
the known contact plan.
Thus, realized future source arrivals, hidden source queues, future random
bundle-size realizations, RNG state, traffic seeds, and trace identifiers are
not exposed to the policy.

To obtain a fixed-dimensional representation from the variable-cardinality
bundle set, the network forms masked mean and elementwise-maximum summaries
of the active bundle rows. Let
\begin{equation}
\mathcal I^r(s)
\triangleq
\left\{
u\in
\widetilde{\mathcal A}^{r}(s)\setminus\{\mathrm{wait}\}
:
\chi^r(u;s)=1
\right\}
\label{eq:pooling_support}
\vspace{-1mm}
\end{equation}
denote the active bundle support for control stage $r$. Define
\begin{equation}
\begin{aligned}
\bar x^r(s)
&=
\operatorname{MeanPool}
\left(
\{x_u(s):u\in\mathcal I^r(s)\}
\right),\\
x_{\max}^r(s)
&=
\operatorname{MaxPool}
\left(
\{x_u(s):u\in\mathcal I^r(s)\}
\right).
\end{aligned}
\label{eq:set_pooling}
\end{equation}
The raw 50-dimensional global context is then concatenated with the two
18-dimensional pooled summaries:
\begin{equation}
z^r(s)
\triangleq
\left[
x_g(s)^{\mathsf T},
\bar x^r(s)^{\mathsf T},
x_{\max}^r(s)^{\mathsf T}
\right]^{\mathsf T}
\in
\mathbb R^{86}.
\label{eq:pooled_context}
\end{equation}
Hence, the 50-dimensional raw global observation $x_g(s)$ and the
86-dimensional pooled network context $z^r(s)$ are distinct quantities. For value estimation, we use the same mean--max pooling operation over all
bundles represented in the relay-side control state, irrespective of their
immediate action feasibility. Define
\vspace{-3mm}
\begin{equation}
\mathcal I^V(s)
\triangleq
\mathcal B(s)
\cup
\{b^{\mathrm{new}}(s):b^{\mathrm{new}}(s)\neq\emptyset\},
\vspace{-3mm}
\end{equation}
and
\begin{equation}
\begin{aligned}[t]
&z^V(s)
\triangleq
\Big[
x_g(s)^{\mathsf T},\\
&\quad\operatorname{MeanPool}
\{x_u(s):u\in\mathcal I^V(s)\}^{\mathsf T},\\
&\quad\operatorname{MaxPool}
\{x_u(s):u\in\mathcal I^V(s)\}^{\mathsf T}
\Big]^{\mathsf T}
\in\mathbb R^{86}.
\end{aligned}
\label{eq:critic_context}
\vspace{-2mm}
\end{equation}
Empty bundle sets use zero pooled summaries. For a represented bundle action $u$, the corresponding bundle row is
concatenated with the pooled context,
\begin{equation}
\widetilde x_u^r(s)
\triangleq
\left[
x_u(s)^{\mathsf T},
z^r(s)^{\mathsf T}
\right]^{\mathsf T}
\in
\mathbb R^{104}.
\label{eq:candidate_scorer_input}
\vspace{-2mm}
\end{equation}
Separate candidate-scoring networks are used for drop and scheduling stages.
Thus, for a represented bundle,
\begin{equation}
\ell_\theta^r(u;o(s))
=
f_\theta^r
\left(
\widetilde x_u^r(s)
\right),~r\in\{\mathrm{drop},\mathrm{sched}\},
\label{eq:bundle_action_logit}
\vspace{-1mm}
\end{equation}
whereas the elective wait score at a scheduling stage is obtained directly
from the pooled context,
\begin{equation}
\ell_\theta^{\mathrm{sched}}
(\mathrm{wait};o(s))
=
f_\theta^{\mathrm{wait}}
\left(
z^{\mathrm{sched}}(s)
\right).
\label{eq:wait_action_logit}
\vspace{-1mm}
\end{equation}
Feasibility masking is applied to the logits before normalization:
\vspace{-1mm}
\begin{gather}
\pi_\theta^r(u|o(s))
=
\frac{
\exp\!\left(\ell_\theta^r(u;o(s))\right)
\chi^r(u;s)
}{
\displaystyle
\sum_{u'\in\widetilde{\mathcal A}^{r}(s)}
\exp\!\left(\ell_\theta^r(u';o(s))\right)
\chi^r(u';s)
},
\raisetag{12pt}
\label{eq:stage_masked_policy}
\vspace{-3mm}
\end{gather}
for $r\in\{\mathrm{drop},\mathrm{sched}\}$. Consequently, infeasible nominal
actions receive zero probability and contribute neither to action selection
nor to the stage entropy.

\subsection{Compound Transition Policy and Value Estimates}
\label{subsec:compound_policy}
Let
$o_k^{d,j}=o(s_k^{d,j})$ denote the observation at drop stage $j$, and let
$o_k^q=o(s_k^d)$ denote the post-overflow scheduling observation. Define
$\mathbb I_k^q=1$ when a scheduling stage is active and
$\mathbb I_k^q=0$ otherwise.

The stage-wise policies induce the following probability for the complete
physical-transition action
$a_k=(\mathbf d_k,q_k)$ conditional on the initial relay-side state $s_k$:
\begin{equation}
\begin{aligned}
\pi_\theta(a_k|s_k)
\triangleq{}&
\prod_{j=0}^{J_k-1}
\pi_\theta^{\mathrm{drop}}
\!\left(
d_{k,j}\mid o_k^{d,j}
\right)
\\
&\times
\left[
\pi_\theta^{\mathrm{sched}}
\!\left(
q_k\mid o_k^q
\right)
\right]^{\mathbb I_k^q}.
\end{aligned}
\label{eq:masked_policy}
\end{equation}
An empty drop product equals one. Accordingly,
\vspace{-2mm}
\begin{equation}
\begin{aligned}
\log\pi_\theta(a_k|s_k)
={}&
\sum_{j=0}^{J_k-1}
\log
\pi_\theta^{\mathrm{drop}}
\!\left(
d_{k,j}\mid o_k^{d,j}
\right)
\\
&+
\mathbb I_k^q
\log
\pi_\theta^{\mathrm{sched}}
\!\left(
q_k\mid o_k^q
\right).
\end{aligned}
\label{eq:compound_log_probability}
\end{equation}

Equation~\eqref{eq:masked_policy} subsumes all control cases. If no overflow
resolution is required, the empty drop product equals one. If scheduling is
enabled, the scheduling probability is multiplied by the probabilities of
all preceding same-time drop decisions; otherwise only the drop-stage
probabilities contribute. If neither a drop nor a scheduling decision is
required, the complete action probability equals one. Thus, all same-time
control stages belong to one physical-transition policy probability.

The intermediate drop-stage states
$s_k^{d,0},\ldots,s_k^{d,J_k}$ and their policy observations
are used only for same-time conditional action selection; they do not
constitute separate physical SMDP transitions, holding-time intervals,
rewards, constraint costs, or critic time steps.

At physical transition boundary $k$, define
$
z_k^V\triangleq z^V(s_k).
$
The critic therefore remains well defined whether the current epoch contains
a drop stage, a scheduling stage, or no active control stage. It produces one
reward value and one value for each constraint:
\begin{equation}
V_R(z_k^V),
\qquad
V_{C_j}(z_k^V),
\quad j\in\mathcal J.
\label{eq:reward_cost_value_functions}
\vspace{-1mm}
\end{equation}
Critic evaluation occurs only at physical SMDP transition boundaries, not at
intermediate same-time drop stages.
\vspace{-1mm}
\subsection{Holding-Time-Aware Lagrangian PPO}
 To account separately for irregular
holding times in temporal credit assignment, define
\begin{equation}
\lambda_k^{\mathrm{time}}
\triangleq
\exp\!\left(
-\frac{\delta_k}{\tau_{\mathrm{GAE}}}
\right),
\label{eq:time_gae_decay}
\end{equation}
where $\tau_{\mathrm{GAE}}>0$ is the physical-time scale used by the
advantage estimator.

For numerical conditioning, rewards and constraint costs are rescaled as
\begin{equation}
\widetilde r_k=\frac{r_k}{s_R},
\qquad
\widetilde c_k^{(j)}=\frac{c_k^{(j)}}{s_j},
\label{eq:optimizer_scaling}
\vspace{-1mm}
\end{equation}
where $s_R>0$ and $s_j>0$ are fixed optimization scales.
These rescalings affect only numerical conditioning and not the physical
objective or constraint definitions.

The reward and constraint TD residuals are
\begin{equation}
\begin{aligned}
\zeta_k^R
&\triangleq
\widetilde r_k
+
V_R(z_{k+1}^V)
-
V_R(z_k^V),\\
\zeta_k^{C_j}
&\triangleq
\widetilde c_k^{(j)}
+
V_{C_j}(z_{k+1}^V)
-
V_{C_j}(z_k^V).
\end{aligned}
\vspace{-1mm}
\end{equation}
The holding-time-aware GAE recursions
\cite{schulman2015gae} are
\begin{equation}
\begin{aligned}
\hat A_k^R
&=
\zeta_k^R
+
\lambda_k^{\mathrm{time}}
\hat A_{k+1}^R,
\\
\hat A_k^{C_j}
&=
\zeta_k^{C_j}
+
\lambda_k^{\mathrm{time}}
\hat A_{k+1}^{C_j}.
\end{aligned}
\label{eq:reward_cost_gae}
\end{equation}
Terminal transitions use zero bootstrap values. If a rollout segment is
truncated before termination, its final critic estimates may instead be used
for bootstrap.

The episodic reward and constraint objectives remain undiscounted;
$\lambda_k^{\mathrm{time}}$ affects only temporal credit assignment in GAE.

Let $\boldsymbol{\lambda}$ denote the current Lagrange multipliers. The
Lagrangian advantage is
\begin{equation}
\hat A_k^L
\triangleq
\hat A_k^R
-
\sum_{j\in\mathcal J}
\lambda_j
\hat A_k^{C_j}.
\label{eq:lagrangian_advantage}
\end{equation}
Before the PPO actor update, $\hat A_k^L$ is standardized over transitions
with an active control stage; the individual reward and constraint advantages
are not separately standardized. For notational simplicity, we retain
$\hat A_k^L$ for the standardized quantity thereafter.
Let $\theta_{\mathrm{old}}$ denote the behavior policy that generated the
rollout. The PPO importance ratio for physical transition $k$ is
\vspace{-1mm}
\begin{equation}
\xi_k(\theta)
\triangleq
\frac{
\pi_\theta(a_k|s_k)
}{
\pi_{\theta_{\mathrm{old}}}(a_k|s_k)
}.
\label{eq:ppo_ratio}
\vspace{-1mm}
\end{equation}
In implementation, the behavior log-probability is stored when the transition
is collected, and the numerator is reconstructed during optimization by
summing the current log-probabilities of the same active control stages in
Equation~\eqref{eq:compound_log_probability}. Stage observations, stable
candidate identities, nominal action slates, and feasibility masks are stored
with the rollout so that the behavior and current policies are evaluated on
identical action supports.

The clipped PPO actor objective is
\begin{equation}
\begin{aligned}
&J_{\mathrm{actor}}(\theta)
=
\mathbb E\Big[
\min\Big(
\xi_k(\theta)\hat A_k^L,
\\
&\quad\mathrm{clip}
\big(
\xi_k(\theta),
1-\epsilon_{\mathrm{ppo}},
1+\epsilon_{\mathrm{ppo}}
\big)
\hat A_k^L
\Big)
\Big]\\
&\quad+
\eta_H\mathbb E[\mathcal H_k],
\end{aligned}
\label{eq:ppo_actor_objective}
\end{equation}
where $\eta_H\geq0$ is the entropy coefficient. For physical transition $k$,
the entropy contribution is the sum of the entropies of all active same-time
control stages:
\begin{equation}
\mathcal H_k
=
\sum_{j=0}^{J_k-1}
\mathcal H^{\mathrm{drop}}(o_k^{d,j})
+
\mathbb I_k^q
\mathcal H^{\mathrm{sched}}(o_k^q),
\label{eq:transition_entropy}
\vspace{-3mm}
\end{equation}
with
\begin{equation}
\mathcal H^r(o(s))
=
-
\sum_{u\in\mathcal A^r(s)}
\pi_\theta^r(u|o(s))
\log
\pi_\theta^r(u|o(s)).
\label{eq:masked_stage_entropy}
\end{equation}
Thus, infeasible nominal actions contribute neither probability mass nor
entropy.

Using critic targets
\begin{equation}
\begin{aligned}
\hat R_k
&=
\hat A_k^R
+
V_R^{\mathrm{old}}(z_k^V),\\
\hat C_{j,k}
&=
\hat A_k^{C_j}
+
V_{C_j}^{\mathrm{old}}(z_k^V),
\end{aligned}
\vspace{-1mm}
\label{eq:critic_targets}
\end{equation}
the joint loss for the five critic outputs is
\vspace{-1mm}
\begin{equation}
\begin{aligned}
\mathcal L_{\mathrm{critic}}
&=
\mathbb E
\Big[
\big(
V_R(z_k^V)-\hat R_k
\big)^2\\
&\quad+
\sum_{j\in\mathcal J}
\big(
V_{C_j}(z_k^V)-\hat C_{j,k}
\big)^2
\Big].
\end{aligned}
\vspace{-1mm}
\label{eq:critic_loss}
\end{equation}
Actor and critic parameters are updated from the same physical-transition
minibatches using their respective optimization losses and learning-rate
schedules. Numerical optimizer settings are specified with the experimental
training configuration rather than being part of the constrained SMDP
definition.
\vspace{-2mm}
\subsection{Lagrange-Multiplier Update}

Multiplier updates use completed episodic constraint returns. Let
$\mathcal E_{\mathrm{comp}}$ denote the set of episodes completed since the
previous multiplier update. The empirical return of constraint $j$ is
\begin{equation}
\widehat J_{C_j}
\triangleq
\frac{1}{|\mathcal E_{\mathrm{comp}}|}
\sum_{e\in\mathcal E_{\mathrm{comp}}}
\sum_{k=0}^{K_e-1}
c_{e,k}^{(j)},
\label{eq:empirical_constraint_return}
\vspace{-1mm}
\end{equation}
The projected dual update is
\begin{equation}
\lambda_j
\leftarrow
\left[
\lambda_j
+
\eta_\lambda
\frac{
\widehat J_{C_j}-d_j
}{
s_j
}
\right]^+,
\qquad
j\in\mathcal J,
\label{eq:lagrange_multiplier_update}
\vspace{-1mm}
\end{equation}
where $\eta_\lambda>0$ is the dual step size. Thus, an episodic return above
its budget increases the corresponding multiplier, while a return below its
budget permits the multiplier to decrease toward zero. The normalization by $s_j$ affects only the numerical dual update and does
not change the physical constraint $J_{C_j}\leq d_j$.
\vspace{-1mm}
\subsection{Training Procedure}
\label{subsec:training_procedure}

Each rollout record corresponds to one physical SMDP transition and stores
\begin{equation}
\left(
o_k,
a_k,
r_k,
\{c_k^{(j)}\}_{j\in\mathcal J},
\delta_k,
o_{k+1},
\mathcal C_k,
\iota_k^{\mathrm{term}}
\right),
\label{eq:rollout_record}
\end{equation}
where $\mathcal C_k$ contains the same-time control-stage observations,
stable candidate identities, nominal action slates, feasibility masks,
selected actions, and behavior log-probabilities required to reconstruct the
compound action probability. The indicator
$\iota_k^{\mathrm{term}}\in\{0,1\}$ identifies terminal physical
transitions. The pooled contexts $z_k^V$ and $z_{k+1}^V$ are deterministic
functions of the stored observations and masks and are reconstructed by the
network when value estimates are evaluated.

After rollout collection, the holding-time factors, reward and constraint
advantages, critic targets, and Lagrangian advantages are computed from the
physical-transition sequence. The quantities used as PPO targets, together
with the stored behavior log-probabilities, remain fixed during the PPO
epochs for that rollout. Current policy probabilities, masked entropies, and
value estimates are re-evaluated on each optimization minibatch.

The actor and critics are trained from random initialization using PPO. Stage-specific
action masks are applied during both training and execution. Multiplier
updates use completed episodic returns and are performed after the PPO updates
for the corresponding rollout, so the updated multipliers affect subsequent
data collection rather than changing advantages already computed for the
current rollout. Algorithm~\ref{alg:masked_lagrangian_ppo} summarizes the procedure.

\begin{algorithm}[t]
\footnotesize
\caption{Masked Lagrangian PPO for Event-Driven Relay Control}
\label{alg:masked_lagrangian_ppo}
\begin{algorithmic}[1]

\REQUIRE
$T,\tau_{\mathrm{GAE}},\epsilon_{\mathrm{ppo}},
\eta_H,\eta_\lambda,\{d_j\}_{j\in\mathcal J},K_{\mathrm{PPO}}$

\STATE Initialize actor, critics, and $\lambda_j\geq0$

\WHILE{training criterion not satisfied}

    \STATE Initialize rollout buffer $\mathcal R$ and
    completed-episode set $\mathcal E_{\mathrm{comp}}$

    \WHILE{rollout collection is incomplete}

        \STATE Observe $s_k$ and construct $o_k$;
        initialize $\mathbf d_k\leftarrow()$

        \WHILE{overflow resolution is required}
            \STATE Construct masked drop policy at $s_k^{d,j}$
            and sample $d_{k,j}$
            \STATE Apply $d_{k,j}$ and append it to $\mathbf d_k$
        \ENDWHILE

        \IF{scheduling is enabled}
            \STATE Construct masked scheduling policy at $s_k^d$
            and sample $q_k$
        \ELSE
            \STATE $q_k\leftarrow\emptyset$
        \ENDIF

        \STATE Set $a_k\leftarrow(\mathbf d_k,q_k)$ and advance DTNSim to $t_{k+1}$

        \STATE Observe
        $(r_k,\{c_k^{(j)}\}_{j\in\mathcal J},\delta_k,o_{k+1})$
        and store the transition and active stage records in $\mathcal R$

        \STATE Accumulate completed episodic constraint returns

    \ENDWHILE

    \STATE Compute
    $\lambda_k^{\mathrm{time}}
    =\exp(-\delta_k/\tau_{\mathrm{GAE}})$,
    reward/constraint GAE, critic targets, and $\hat A_k^L$

    \FOR{$p=1,\ldots,K_{\mathrm{PPO}}$}
        \FOR{each shuffled minibatch $\mathcal M\subset\mathcal R$}
            \STATE Re-evaluate masked stage probabilities and form PPO ratios
            \STATE Update actor and critic parameters
        \ENDFOR
    \ENDFOR

    \IF{$\mathcal E_{\mathrm{comp}}\neq\emptyset$}
        \STATE Estimate $\widehat J_{C_j}$ and update $\lambda_j$
    \ENDIF
\ENDWHILE
\end{algorithmic}
\end{algorithm}
\section{DTNSim Implementation and Simulator Interface}
\label{sec:dtnsim_implementation}

We extend DTNSim~\cite{msancheznet2019dtnsim} to implement the
event-driven relay environment of Section~\ref{sec:event_driven_smdp}
and the masked control interface of
Section~\ref{sec:masked_lagrangian_ppo}. The implementation preserves the
shared-buffer, non-preemptive service, delayed-reception, and co-timed event
semantics defined by the model.

\subsection{Runtime Architecture and Environment Contract}
DTNSim owns the simulation clock, source generation and source queues,
source--relay transmission and propagation, contact plan, bundle movement,
and physical event calendar. A relay-control adapter exposes the relay-side
state, constructs the causal observation and stage-specific masks, applies
relay actions, and returns decision-to-decision SMDP transitions. The training
driver updates the actor, critic parameters, and Lagrange multipliers without
directly modifying simulator state.

Each decision epoch begins after all physical events at time $t_k$ have been
processed. If relay capacity is exceeded, the controller repeatedly selects
same-time rejection or eviction actions
$\mathbf d_k=(d_{k,0},\ldots,d_{k,J_k-1})$
until the overflow is resolved, recomputing the observation and mask after
each decision. Scheduling is then reconsidered on the resulting state. The
complete relay action is therefore $a_k=(\mathbf d_k,q_k)$,
with no physical time elapsing between its internal drop and scheduling
stages.

After the compound action is completed, DTNSim advances to the next
decision epoch. Following an elective wait while a feasible transmission
exists, the controller is awakened after at most $30$~s or earlier if a
physical event occurs. The interface returns the next observation,
reward $r_k$, constraint-cost vector, holding time $\delta_k$, termination
status, and the stage information required to reconstruct the compound
policy probability during PPO optimization.

At the terminal time $T$, physical events scheduled exactly at $T$ are
processed before termination. The simulator stops without draining the
network state; unexpired backlog is unresolved rather than counted
as loss. Learned and heuristic controllers use the same physical simulator,
contact plan, traffic realization, and relay-control interface.
\subsection{DTNSim Extensions and Validation}

\vspace{1mm}
\textbf{Shared relay storage:}
Moon and Mars bundles share the finite relay-buffer budget $B_{\max}$.
Bundles already in relay--Earth service or propagation are excluded from
buffer occupancy. DTNSim applies no independent overflow rule; all capacity
violations are resolved through the controller-defined rejection/eviction
sequence.

\vspace{1mm}
\textbf{Contact-aware service:}
Relay--Earth transmission is non-preemptive and single-server. A bundle may
enter service only if its complete serialization fits within the active
contact and its Earth reception occurs strictly before absolute expiry, i.e.,
$t_k+L(b)/R_{RE}\leq t_{\mathrm{end}}$ and
$t_k+L(b)/R_{RE}+\tau_{RE}<t_b^{\mathrm{exp}}$.
Thus, service admission requires timely Earth reception, not merely
transmission completion before contact closure. A buffered bundle
that can no longer satisfy the delivery deadline remains stored until overflow
removal or absolute expiry, but is masked from relay--Earth scheduling.

\vspace{1mm}
\textbf{Delayed reception and event ordering:}
Relay transmission completion releases the server but does not update
destination freshness; the corresponding update occurs only at Earth
reception. Co-timed events follow the deterministic ordering in
Subsection~\ref{subsec:event_types_ordering}, after which all applicable
same-time control stages are completed before physical time advances.

\vspace{1mm}
The frozen DTNSim integration was validated against the Python reference
implementation before policy evaluation. Source-ingress equivalence was
verified over 67 traces and 8,085 bundles with zero timing deviation, and
134 automated tests covered observations, masks, event handling, and metric
computation. Additional regression tests verified shared-buffer accounting,
action feasibility, strict Earth-before-expiry delivery, delayed reception,
co-timed event ordering, and reproducibility under fixed traces and frozen
checkpoints.
After validation, the held-out evaluation campaign compared five controllers
on 64 common traffic traces, yielding 320 evaluation episodes under identical
simulator conditions.

\section{Evaluation Setup}
\label{sec:evaluation_setup}

All policies are evaluated through the common DTNSim execution path of
Section~\ref{sec:dtnsim_implementation}. They use identical traffic
realizations, contact opportunities, relay capacity, expiry semantics, and
physical feasibility rules, and differ only in relay-side action selection.
Table~\ref{tab:ppo_training_parameters} summarizes the frozen simulation,
training, and evaluation configuration.
\vspace{-2mm}
\subsection{Scenario and Training Configuration}

The experiments use the Moon--Mars--relay--Earth topology of
Section~\ref{sec:system}. Source generation, source-side FIFO service,
propagation, shared-buffer admission, relay--Earth transmission, and delayed
Earth reception follow the system model and are not repeated.
The same contact-plan freshness reference is used for all policies.

\begin{table}[t]
\centering
\caption{Frozen evaluation and PPO training configuration.}
\vspace{-2mm}
\label{tab:ppo_training_parameters}
\footnotesize
\setlength{\tabcolsep}{2pt}
\renewcommand{\arraystretch}{1.08}
\begin{tabular}{@{}>{\raggedright\arraybackslash}p{0.37\linewidth}
>{\raggedright\arraybackslash}p{0.61\linewidth}@{}}
\toprule
\textbf{Parameter} & \textbf{Setting} \\
\midrule

Episode interval &
$t_0=1800$~s, $T=16200$~s; $H=14400$~s \\

Traffic generation &
Independent Poisson; mean interarrival:
Moon $180$~s, Mars $360$~s \\

Bundle sizes &
$20/100/300$~Mbit with probabilities $0.5/0.3/0.2$ \\

Bundle TTL &
$5400$~s from original generation \\

Relay capacity $B_{\max}$ &
$1.2\times10^9$~bits \\

Contact cycle $T_c$ &
$1800$~s \\

Contact windows &
$M\!\to\!R:[0,1200)$;
$S\!\to\!R:[900,1500)$;
$R\!\to\!E:[1200,1800)$ \\

Link rates &
$2$~Mbit/s on all links \\

Propagation delays &
$\tau_{MR}=1.3$~s;
$\tau_{SR}=600$~s;
$\tau_{RE}=0.12$~s \\

Reference bundle size &
$L^{\mathrm{ref}}=20$~Mbit \\

Freshness objective &
$w_M=w_S=1$; $\phi_M(x)=\phi_S(x)=x$ \\

Initial Earth AoI &
$\Delta_M(t_0)=1190.0$~s;
$\Delta_S(t_0)=1179.88$~s \\

Fairness rates &
$\rho_M=0.8441$; $\rho_S=0.8364$ \\

GAE time scale &
$\tau_{\mathrm{GAE}}=900$~s \\

PPO clipping &
$\epsilon_{\mathrm{ppo}}=0.2$ \\

Actor learning rate &
$10^{-4}\!\rightarrow10^{-5}$ (linear decay) \\

Critic learning rate &
$10^{-3}\!\rightarrow10^{-4}$ (linear decay) \\

Entropy coefficient &
$0.01\!\rightarrow0.001$ (linear decay) \\

PPO batching &
$16$ episodes/update; $4$ PPO epochs;
$512$ physical transitions/minibatch \\

Training duration &
$2048$ episodes; random initialization \\

Dual update &
$\eta_\lambda=0.05$;
$\boldsymbol{\lambda}_0=(0.1,0.1,0,0)$ \\

Reward/cost scales &
$s_R=HT_c$;
$(s_{\mathrm{drop}},s_{\mathrm{exp}},
s_{\mathrm{fair},M},s_{\mathrm{fair},S})
=(17,1,H,H)$ \\

Optimization safeguards &
max.\ gradient norm $0.5$; KL threshold $0.03$ \\

Constraint budgets &
$(d_{\mathrm{drop}},d_{\mathrm{exp}},
d_{\mathrm{fair},M},d_{\mathrm{fair},S})
=(17,1,0,0)$ \\

Held-out evaluation &
$5$ controllers $\times$ $64$ common traffic traces
($320$ evaluation episodes) \\

\bottomrule
\end{tabular}
\vspace{-4mm}
\end{table}

\vspace{-2mm}
\subsection{Compared Policies}
\label{subsec:compared_policies}

We compare the proposed \textbf{Masked Lagrangian PPO} controller with four
deterministic baselines. \textbf{FIFO} schedules the feasible bundle with the
earliest relay arrival, \textbf{LIFO} the latest relay arrival,
\textbf{EDF} the earliest absolute expiry, and \textbf{Freshest-first} the
largest generation timestamp. During overflow, each baseline discards its
least-preferred bundle under the corresponding ordering: FIFO discards the
latest arrival, LIFO the earliest arrival, EDF the latest expiry, and
Freshest-first the smallest generation timestamp. Remaining ties are resolved
using the stable bundle identifier. All policies use the same candidate sets
and physical feasibility rules.
\vspace{-2mm}
\subsection{Evaluation Protocol and Metrics}

Each policy is evaluated on the same $64$ held-out stochastic traffic traces,
ensuring identical exogenous conditions across policy comparisons. Reported
performance values are arithmetic means over these common traces. All training
hyperparameters and evaluation rules were fixed before the held-out evaluation.

The primary performance metric is the time-average aggregate excess-age cost
$\bar G\triangleq
\frac{1}{H}\int_{t_0}^{T}G(t)\,dt$,
where lower values indicate less excess staleness relative to the
contact-plan reference.
With unit source weights and linear penalties, this equals the time-average
aggregate excess AoI used for policy comparison.

Secondary diagnostics include source-specific excess, drop and
expiry counts, delivered traffic, relay--Earth contact utilization,
cumulative useful freshness progress
$\sum_{k=0}^{K_T-1}y_{i,k}$, terminal deficits $D_i(K_T)$, and episodic
constraint returns. Relay--Earth contact utilization is defined as transmitted
bits divided by the total available relay--Earth contact capacity. Constraint
compliance is assessed against the budgets in
Table~\ref{tab:ppo_training_parameters}.
\vspace{-2mm}
\section{Results and Discussion}
\label{sec:results}

We compare Masked Lagrangian PPO with FIFO, LIFO, EDF, and Freshest-first
under the common evaluation protocol of Section~\ref{sec:evaluation_setup}.
All values are computed over the same $64$ held-out stochastic traffic traces.
\vspace{-2mm}
\subsection{Freshness Performance}

Table~\ref{tab:freshness_results} reports the time-average aggregate and
source-specific excess AoI. Masked Lagrangian PPO achieves the lowest
aggregate excess AoI among the evaluated policies.

\begin{table}[t]
\centering
\caption{Excess-AoI performance over the $64$ held-out
traces. Values are means in seconds; lower is better.}
\vspace{-1.5mm}
\label{tab:freshness_results}
\footnotesize
\setlength{\tabcolsep}{8pt}
\renewcommand{\arraystretch}{1.05}
\begin{tabular}{lccc}
\toprule
\textbf{Policy} &
\textbf{Aggregate} $\downarrow$ &
\textbf{Moon} $\downarrow$ &
\textbf{Mars} $\downarrow$ \\
\midrule
Masked Lag. PPO &
$\mathbf{887.63}$ &
$\mathbf{216.64}$ &
$\mathbf{670.98}$ \\

FIFO &
$1747.87$ &
$686.89$ &
$1060.97$ \\

LIFO &
$993.57$ &
$232.83$ &
$760.74$ \\

EDF &
$1692.38$ &
$806.15$ &
$886.23$ \\

Freshest-first &
$1009.56$ &
$232.82$ &
$776.74$ \\
\bottomrule
\end{tabular}
\vspace{-1mm}
\end{table}

Relative to FIFO and EDF, the controller reduces mean aggregate
excess AoI by $49.2\%$ and $47.6\%$, respectively. Relative to the
LIFO and Freshest-first baselines, the reductions are $10.7\%$ and $12.1\%$.
The improvement is present for both Moon and Mars, showing that the
aggregate result is not obtained by sacrificing one source. Because all policies use common traffic traces, performance can be
compared trace by trace. Masked Lagrangian PPO has lower aggregate excess AoI
than FIFO and EDF on all $64$ traces and lower excess AoI than both LIFO and
Freshest-first on $59$ of the $64$ traces. Thus, the mean advantage is broadly
distributed across the evaluation set rather than driven by a small
number of favorable traces.
\vspace{-3mm}
\subsection{Operational Performance}

Table~\ref{tab:operational_results} summarizes relay operation under the
heterogeneous bundle-size distribution. The bundle-based and bit-based
delivery metrics are both reported because they need not rank policies
identically when bundle sizes differ.

\begin{table*}[t]
\centering
\caption{Mean operational performance over $64$ traces.
Expiries count relay-side events used in the constraint.}
\vspace{-1mm}
\label{tab:operational_results}
\footnotesize
\setlength{\tabcolsep}{4pt}
\renewcommand{\arraystretch}{1.05}
\begin{tabular}{lccccccc}
\toprule
\textbf{Policy} &
\textbf{Drops} $\downarrow$ &
\textbf{Exp.} $\downarrow$ &
\textbf{Delivered} &
\textbf{Delivered Gbit} &
\textbf{Bundle DR} &
\textbf{Bit DR} &
$U_{\mathrm{contact}}$ \\
\midrule
Masked Lag. PPO &
\textbf{10.02} & 0.05 & \textbf{103.72} & 8.086 &
\textbf{0.860} & 0.686 & 0.845 \\

FIFO &
16.20 & \textbf{0.00} & 96.28 & \textbf{8.578} &
0.801 & \textbf{0.727} & \textbf{0.897} \\

LIFO &
17.88 & 0.20 & 95.42 & 8.286 &
0.794 & 0.703 & 0.865 \\

EDF &
16.91 & \textbf{0.00} & 95.61 & 8.577 &
0.796 & \textbf{0.727} & 0.896 \\

Freshest-first &
16.78 & 0.17 & 96.56 & 8.262 &
0.803 & 0.701 & 0.863 \\
\bottomrule
\end{tabular}
\vspace{-3mm}
\end{table*}

The proposed controller produces substantially fewer overflow drops and
delivers the largest number of bundles, giving the highest bundle delivery
ratio. However, it delivers fewer total bits and uses less relay--Earth
contact capacity than the baselines. Thus, its freshness improvement is not
obtained by transmitting more data or occupying more downlink capacity.
Instead, the results indicate more selective use of the available transmission
opportunities according to information freshness.

The distinction between bundle and bit delivery is important in this setting:
Masked Lagrangian PPO achieves a bundle delivery ratio of $0.860$ while its
bit delivery ratio is $0.686$. This difference reflects the heterogeneous
bundle sizes and prevents bundle-count delivery alone from being interpreted
as a throughput metric.

\begin{figure}
\centering
    \includegraphics[width=0.5\textwidth]
    {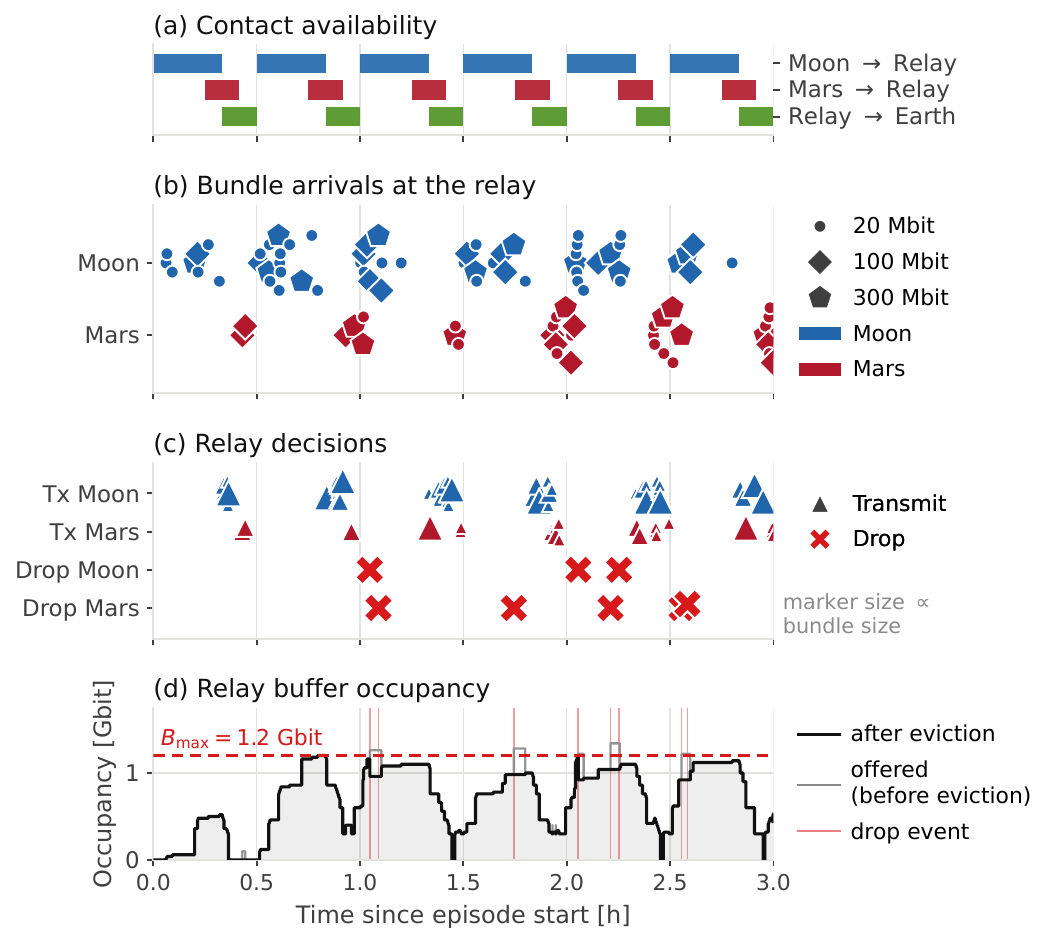}
    \vspace{-6mm}
    \caption{Relay operation under the PPO policy.}
    \label{fig:contact_decision_timeline}
    \vspace{-3mm}
\end{figure}

Figure~\ref{fig:contact_decision_timeline} provides an event-level view of
the relay behavior underlying the aggregate results in Table~\ref{tab:operational_results}.
Over the representative three-hour interval, Moon and Mars bundles arrive
asynchronously with heterogeneous sizes of 20, 100, and 300~Mbit, while
their transmission opportunities are restricted by the scheduled contacts.
The relay consequently selects different buffered bundles for transmission
and invokes drop decisions when incoming traffic creates buffer pressure.
As seen in Fig.~\ref{fig:contact_decision_timeline}(d), the offered
pre-eviction occupancy can exceed the $B_{\max}=1.2$~Gbit capacity, whereas
the post-decision occupancy is maintained within the buffer limit. The
timeline therefore illustrates how scheduling and overflow-management
decisions interact with intermittent connectivity, heterogeneous bundle
sizes, and finite relay storage.

\subsection{Fairness and Constraint Compliance}
Table~\ref{tab:constraint_results} reports the empirical mean episodic
constraint returns. The corresponding budgets are
$(17,1,0,0)$.

\begin{table*}[t]
\centering
\caption{Mean episodic constraint returns over the $64$ held-out traces.
A policy is mean-feasible when all four empirical means satisfy their
corresponding budgets.}
\vspace{-1.5mm}
\label{tab:constraint_results}
\footnotesize
\setlength{\tabcolsep}{6pt}
\renewcommand{\arraystretch}{1.05}
\begin{tabular}{lccccc}
\toprule
\textbf{Policy} &
$\widehat J_{C_{\mathrm{drop}}}$ &
$\widehat J_{C_{\mathrm{exp}}}$ &
$\widehat J_{C_{\mathrm{fair},M}}$ &
$\widehat J_{C_{\mathrm{fair},S}}$ &
\textbf{Mean-feasible} \\
\midrule
\textbf{Budget} &
$17$ & $1$ & $0$ & $0$ & --- \\
\midrule
Masked Lag. PPO &
10.02 & 0.05 & $-2007.89$ & $-1693.17$ & Yes \\

FIFO &
16.20 & 0.00 & $-1764.52$ & $-1138.05$ & Yes \\

LIFO &
17.88 & 0.20 & $-2010.25$ & $-1659.19$ & No \\

EDF &
16.91 & 0.00 & $-1592.67$ & $-1452.85$ & Yes \\

Freshest-first &
16.78 & 0.17 & $-2010.25$ & $-1658.08$ & Yes \\
\bottomrule
\end{tabular}
\vspace{-2mm}
\end{table*}

Masked Lagrangian PPO satisfies all four budgets on the empirical mean.
LIFO is the only selected baseline that is not mean-feasible, because its
mean drop return $17.88$ exceeds the budget of $17$. All five policies have
negative mean source-specific fairness returns, indicating that their mean
delivered freshness progress exceeds the required episodic service level for
both sources.

For additional finite-horizon context, the proposed controller achieves
cumulative useful freshness progress of $14162.9$~s for Moon and $13737.8$~s
for Mars, with terminal reflected deficits of $484.0$~s and $237.2$~s,
respectively. These reflected deficits are diagnostic state variables rather
than constraint returns; positive terminal values therefore do not imply
fairness-constraint violation and do not establish or refute the long-run
rate-stability condition of Lemma~\ref{lem:deficit_stability_fairness}.
\subsection{Overall Discussion}
The results show a consistent trade-off between freshness and transmitted
volume. Masked Lagrangian PPO achieves the lowest aggregate and per-source
excess AoI while satisfying the empirical mean operational and fairness
budgets. Its advantage is observed on every trace against FIFO and EDF and on
$59$ of $64$ traces against LIFO and Freshest-first. At the same time, the
controller transmits fewer bits and uses less relay--Earth contact capacity
than the deterministic baselines.

Accordingly, the observed freshness improvement should be interpreted as an
improvement in relay-side information selection rather than increased
throughput. The results are descriptive statistics over the fixed held-out
policy evaluation and do not constitute a claim of statistical significance
or robustness to independent retraining.
\vspace{-1mm}
\vspace{-1mm}
\section{Conclusion}
This paper developed a goal-oriented relay-control framework for deep-space
DTNs with stochastic heterogeneous traffic, finite relay storage, intermittent
contacts, and delayed Earth reception. Joint scheduling and dropping are
formulated as a constrained event-driven SMDP and solved using Masked
Lagrangian PPO with feasibility-aware action masking and variable-cardinality
set encoding. The framework is implemented and evaluated in DTNSim.

Across $64$ common held-out traffic traces, the proposed controller achieves
the lowest mean aggregate excess AoI among the evaluated policies:
$887.63$~s, compared with $993.57$~s for LIFO, $1009.56$~s for
Freshest-first, $1692.38$~s for EDF, and $1747.87$~s for FIFO. It also
satisfies the empirical mean drop, expiry, and source-specific fairness
budgets. These results show that relay decisions based on destination-side
freshness and system context can outperform fixed queue-ordering heuristics
under the considered deep-space DTN scenario.

The present study assumes a known contact plan, fixed routing, a single
controlled relay, and non-preemptive no-resume bundle transmission. Future
work will consider component-wise ablation, contact-plan uncertainty,
partial-transfer service, multi-relay control, and validation with operational
DTN implementations.
\vspace{-2mm}
\section*{Acknowledgment}

This work was supported by the ERC Advanced Grant GOSPACE (No.~101122990), TÜBİTAK Grant 123C572 (GOREV), and the 1515 Frontier R\&D Laboratories Support Program for the Türk Telekom 6G R\&D Lab (Project No.~5249902). The views expressed are solely those of the authors and do not necessarily reflect those of the European Union or the European Research Council Executive Agency, neither of which can be held responsible for them. The authors used AI for language refinement and figure polishing. All AI-assisted material was verified by the authors for consistency with the system architecture.

\vspace{-1mm}
{\raggedright\printbibliography%
}

\end{document}